\documentclass[11pt]{article}
\usepackage[totalwidth=13.0cm,totalheight=20.0cm]{geometry}
\usepackage{latexsym,amsthm,amsmath,amssymb,url}
\usepackage[ruled, linesnumbered]{algorithm2e}

\usepackage{tikz}
\usetikzlibrary{decorations.pathreplacing}
\usepackage{authblk}

\newtheorem{lemma}{Lemma}

\newtheorem{theorem}{Theorem}

\title{An improved kernel for the cycle contraction problem}
\author[a]{Bin Sheng \footnote{Corresponding author. Email: shengbinhello@gmail.com }}
\affil[a]{Department of Computer Science, Royal Holloway\\ University of London, Egham, Surrey, TW20 0EX, UK}
\author[b]{Yuefang Sun}
\affil[b]{Department of Mathematics, Shaoxing University\\Shaoxing, Zhejiang 312000, PR China}

\begin{document}
\maketitle

\begin{abstract}
The problem of modifying a given graph to satisfy certain properties has been one of the central topics in parameterized tractability study. In this paper, we study the cycle contraction problem, which makes a graph into a cycle by edge contractions. The problem has been studied {by Belmonte et al. [IPEC 2013]} who obtained a linear kernel with at most $6k+6$ vertices. We provide an improved kernel with at most $5k+4$ vertices for it in this paper.
\end{abstract}

\section{Introduction}

Parameterized computation is a new approach to tackle NP-hard problems, it has successful applications in many fields, {including \textsc{Combinatorial Optimization}, \textsc{Artificial Intelligence}, \textsc{Computational Biology}, and so on}.
A \emph{parameterized problem} is a subset $L\subseteq \Sigma^* \times
\mathbb{N}$ over a finite alphabet $\Sigma$. The problem $L$ is said to be
\emph{fixed-parameter tractable (FPT)} if the membership of its instance
$(x,k)$ in $\Sigma^* \times \mathbb{N}$ can be decided in time
$f(k)|x|^{O(1)}$, where $f$ is a computable function depending on the
{\em parameter} $k$ only.
Given a parameterized problem $L$, a \emph{kernelization of $L$} is a polynomial time
algorithm that shrinks an instance $(x,k)$ into a smaller instance $(x',k')$ (the
\emph{kernel}) such that $(x,k)\in L$ if and only if
$(x',k')\in L$ and $k'+|x'|\leq g(k)$ for some
function $g$.
It is well-known that a decidable parameterized problem $L$ is fixed-parameter
tractable if and only if it has a kernel. Kernels of small size are of the
main research interest, due to application needs. Thus, we have particular interests in kernels whose sizes are bounded by a polynomial function of the parameter. For a more thorough introduction to FPT and Kernelization, we refer the readers to the excellent books \cite{cygan2015parameterized,downey2013fundamentals,Flum:2006:PCT:1121738} and surveys \cite{kratsch2014recent,lokshtanov2012kernelization}.

Graph theory has been a rich source of research problems from the parameterized complexity perspective. Among them, there is a large set of research studying the distance of a graph to a certain property, that is, the minimum number of operations that make the graph satisfy the required property. Most common graph modification operations include deleting (or adding) vertices (or edges). \textsc{Vertex Cover}, \textsc{Feedback Vertex Set}, \textsc{Multiway Cut}, \textsc{Minimum Fill-in} and \textsc{Cluster Editing} are just a few of the extensively studied topics in this research framework. 

Recently, people start to look at the effect of edge contraction on a given graph, and study it in the setting of parameterized tractability.  {The parameterized complexity of the contractibility problem has been investigated for various specific classes of graphs, such as making a graph planar \cite{golovach2013obtaining}, split \cite{guo2015obtaining}, bipartite \cite{guillemot2013faster, heggernes2013obtaining}, or more specifically into a tree or a path \cite{heggernes2014contracting}. We have also seen study of contracting edges to satisfy certain degree bounds \cite{golovach2011increasing, mathieson2012editing} or to eliminate small induced subgraphs \cite{lokshtanov2013hardness}.}

This paper follows this line of research by providing an improvement for the cycle contraction problem, which asks to do minimum number of edge contractions on a given graph and make it into a cycle. The cycle contraction problem has been studied in \cite{belmonte2013parameterized}, where the authors obtained a linear kernel with at most $6k+6$ vertices. We provide an improved kernel for the problem with at most $5k+4$ vertices in this paper.

\section{Notations and Terminology}
For most of the graph theoretical concepts used in this paper, we follow the notations and terminology in \cite{bollobas2013modern}.

An undirected graph is denoted by an ordered pair $G = (V, E)$, where $E$ is a set of
unordered pairs of elements in $V$. The elements of $V$ are the \textit{vertices}
of $G$ and the elements of $E$ are the \textit{edges} of $G$. Two vertices $u, v \in V$ are \textit{adjacent}
if $u \neq v$ and $\{u, v\} \in E$, and we say they are \textit{neighbour} of each other. An edge $\{u, v\}$ is normally written as $uv$ for short, thus $u$ and $v$
are adjacent if and only if $uv \in E$. And in this case we say $u$ is \textit{incident} with the edge $uv$. We denote the \textit{degree} of $u$ in $G$ by $d_{G}(u)$, which is the number of edges incident with $u$.

A graph $H = (V', E')$ is a \textit{subgraph} of $G$ if $V'\subseteq V$ and $E' \subseteq E$. $H = (V', E')$
is an \textit{induced subgraph} of $G$ if $V' \subseteq V$ and $E' = \{uv\in E| u,v \in V' \}$. For a set of vertices $X \subseteq V(G)$, we use $G[X]$ to denote the induced subgraph of $G$ with vertex set $X$.

A \textit{path} is a non-empty graph $P=(V,E)$ with
 vertex set $V=\{ u_0, u_1, \ldots, u_k\}$ and edge set $E=\{u_0u_1, u_1u_2, \ldots, u_{k-1}u_k\}$,
where $u_i$ are all distinct. And we say it is \textit{a path between $u_0$ and $u_k$}, which are called the endvertices of $P$.
If $P=u_0u_1 \ldots u_k$ is a path and $k\geq 2$, then the graph we obtain by adding the edge $u_{k}u_0$ to $P$ is called a \textit{cycle}. The length of a path (or a cycle) is the number of edges in it. A path with at least one edge is called a \textit{nontrivial} path.

A non-empty graph $G$ is \textit{connected} if there is a path between any two of its vertices. A \textit{cut set} in a connected graph is
a set of vertices whose deletion results in a disconnected graph. A connected graph $G$ is said to be \textit{$k$-connected} if every cut set of it has size at least $k$.  {A connected graph $G$ is \textit{$k$-edge-connected} if $G$ remains connected whenever less than $k$ edges are deleted from it.} An edge $e\in E(G)$ in a connected graph $G$ is called a \textit{bridge} if $G-e = (V(G), E(G)-e)$ is disconnected.  A \textit{block} in a graph is a maximal 2-connected subgraph. 

The \emph{contraction} of an edge $uv$ in $G$ removes $u$ and $v$ from $G$, and replaces them by a new vertex adjacent to exactly all the neighbours of $u$ and $v$ in $G$. {Note that, by definition, edge contractions create neither self-loops nor multiple edges.}
 

The following notions come from \cite{belmonte2013parameterized, heggernes2014contracting}.
Let $H$ be a graph. A graph $G$ is \textit{$k$-contractible} to $H$ if we can obtain a copy of $H$ by at most $k$ edge contractions on $G$. And we say $G$ is  {\textit{contractible}} to $H$ if there is some $k$ such that $G$ is $k$-contractible to $H$. The contraction is actually defined by a surjection $\phi: V(G) \rightarrow V(H)$, where $W(h)=\{v \in V(G), \phi(v)= h\}$ is the set of vertices contracted into $h\in V(H)$. The surjection satisfies the following conditions.
\begin{enumerate}
 \item For every vertex $h\in V(H)$, $G[W(h)]$ is a connected subgraph of $G$,
 \item For every pair of vertices in $\{h_i,h_j\} \subseteq V(H)$, $h_{i}h_{j}\in E(H)$ if and only if there is an edge between $W(h_{i})$ and $W(h_{j})$,
 \item $\cup_{h\in V(H)}W(h)=V(G)$ and $W(h_{i}) \cap W(h_{j})=\emptyset$ if $i \neq j$.
\end{enumerate}
We call ${\cal W} =\{W(h), h\in V(H)\}$ an \textit{$H$-witness structure} of $G$. And for each $h\in V(H)$, $W(h)$ is called a \textit{witness set} of $\cal W$.
A witness set $W(h)$ is \textit{big} (\textit{small}, respectively) if $|W(h)|\geq 2$ ($|W(h)| =1$, respectively), and we denote it by $W_{b}(h)$ ($W_{s}(h)$, respectively).

\section{Main Result}

First let us give the formal definition of the parameterized Cycle Contraction problem.
 \begin{quote}
\fbox{~\begin{minipage}{0.9\textwidth}
  {\sc Cycle Contraction \cite{belmonte2013parameterized}} \nopagebreak

    \emph{Instance:} A connected graph $G$ and an integer $k$.
    
    \emph{Parameter:} $k$.

    \emph{Output:} Decide if one can do at most $k$ edge contractions on $G$ to modify it into a cycle.
\end{minipage}~}
  \end{quote}

In this section, we prove that the problem of \textit{Cycle Contraction} admits a kernel with at most $5k+4$ vertices, which is an improvement over the $6k+6$
kernel bound in \cite{belmonte2013parameterized}. Without loss of generality, we assume that the graphs we consider are connected, as there is no way to edge contract a disconnected graph into a cycle or a path. We also assume that the parameter $k\geq 12$ (which implies $5k+4 \geq 4.5k+10$), as for the smaller $k$, the kernel size is at most $4.5k+10$.


In \cite{li2016improved}, the authors study the following parameterized \textit{Path Contraction} problem and obtain a kernel with at most $3k+4$ vertices.

 \begin{quote}
\fbox{~\begin{minipage}{0.9\textwidth}
  {\sc Path Contraction  \cite{li2016improved}} \nopagebreak

    \emph{Instance:} A connected graph $G$ and an integer $k$.
    
    \emph{Parameter:} $k$.

    \emph{Output:} Decide if one can do at most $k$ edge contractions on $G$ to modify it into a path.
\end{minipage}~}
  \end{quote}

\begin{theorem} \label{theorem:PCkernel}\cite{li2016improved}
The parameterized Path Contraction problem admits a kernel with at most $3k + 4$ vertices.
\end{theorem}

 {We add some descriptions of the reduction rules in \cite{li2016improved} here, to help explain how we make use of their result. Their kernel is obtained by exhaustively applying the following two reduction rules. Note that both reduction rules do not decrease the value of the parameter $k$.}

\begin{lemma}\label{lemma:ruleA}\cite{li2016improved}
  For any 2-edge-connected graph $G=(V,E)$, if $G$ is contractible to a path $P$ by $q$ edge contractions, then $q \geq (|V|-1)/3$.
\end{lemma}

The authors implicitly use the following reduction rule, which is implied by Lemma \ref{lemma:ruleA}.
 
 \textbf{Rule A} \cite{li2016improved} If $G$ is a 2-edge-connected graph and $|V(G)| > 3k+1$, then $G$ is a NO instance for Path Contraction with parameter $k$.

 {For a connected graph $G = (V , E )$ , let $C$ be a 2-edge-connected component or a single vertex such that each edge between $V(C)$ and $V\setminus V(C)$ is a bridge. Let $G - V (C) = \{ B_1,\ldots, B_h \}$ be the set of all connected components in
$G[V\setminus V(C)]$, where $h \geq 1$ and $|V(B_i)| \geq |V(B_j)|$ if $1 \leq i < j \leq h$.} 

 {\textbf{Rule B} \cite{li2016improved}\label{rule2} Let $e$ be the bridge of $G$ between $C$ and $B_1$. If $|V \setminus V(B_1)|\geq k+2$ and one of the following inequalities is satisfied: 
\begin{enumerate}
 \item $|V(B_1)|+(|V(C)|-1)/3 \geq k+2$, if $1\leq h \leq 3;$
 \item $|V(B_1)|+(|V(C)|-1)/3 + \Sigma^{h}_{i=4} |V(B_i)|\geq k+2$, if $h \geq 4.$
\end{enumerate}
then return $(G',k)$, where $G'$ is the graph obtained by contracting $e$.}

We will make use of their result to obtain our result on Cycle Contraction. Firstly, we will introduce the problem of \textit{Path Contraction with Fixed Endvertices (PCFE)}, which has the requirement of fixed endvertices.



 \begin{quote}
\fbox{~\begin{minipage}{0.9\textwidth}
  {\sc Path Contraction with Fixed Endvertices (PCFE)} \nopagebreak

    \emph{Instance:} A connected graph $G$, an integer $k$ and two vertices $u, v \in V(G)$.
    
    \emph{Parameter:} $k$.

    \emph{Output:} Decide if one can do at most $k$ edge contractions on $G$ and make it into a path between two vertices $u_0$ and $v_0$, such that $u \in W(u_0)$ and $v \in W(v_0)$.
\end{minipage}~}
  \end{quote}

We will show that PCFE also admits a kernel with at most $3k+4$ vertices. We prove it by reducing an instance of PCFE to an instance $(H, k)$ of the Path Contraction problem.

\begin{theorem}\label{theorem:PCFE}
PCFE admits a kernel with at most $3k+4$ vertices.
\end{theorem}
\begin{proof}
Given an instance $(G, u, v, k)$ of PCFE, where $G$ is a connected graph with $\{u,v\} \subseteq V(G)$. We construct a new graph $H= G + P_1 + P_2$, where $P_1$ is a path with length $k+1$ between $u$ and $u'$, and $P_2$ is a path with length $k+1$ between $v$ and $v'$. An example is shown in Figure \ref{fig:lemmaProof1}. Note that $V(P_1) \cap V(G) = \{u\}$, $V(P_2) \cap V(G) = \{v\}$.

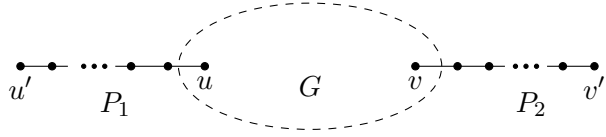
\begin{figure}
\centering
\begin{tikzpicture}[scale = 0.7]

\draw [dashed](5, 0) ellipse (2.5 and 1.2);
\draw [](5,0)node[below]{$G$};
\draw [](1.3,-0.3)node[below]{$P_1$};
\draw [](9.2,-0.3)node[below]{$P_2$};
\draw [](3,0)node[below]{$u$};
\draw [](-0.5,0)node[below]{$u'$};
\draw [](7,0)node[below]{$v$};
\draw [](10.4,0)node[below]{$v'$};

\draw [](3,0)[fill]circle[radius=0.07];
\draw [](2.3,0)[fill]circle[radius=0.07];
\draw [](1.6,0)[fill]circle[radius=0.07];
\draw [](0.1,0)[fill]circle[radius=0.07];
\draw [](-0.5,0)[fill]circle[radius=0.07];

\draw [](1.1,0)[fill]circle[radius=0.04];
\draw [](0.7,0)[fill]circle[radius=0.04];
\draw [](0.9,0)[fill]circle[radius=0.04];

\draw [](9.1,0)[fill]circle[radius=0.04];
\draw [](9.3,0)[fill]circle[radius=0.04];
\draw [](8.9,0)[fill]circle[radius=0.04];

\draw [](-0.5,0)--(0.1,0);
\draw [](2.3,0)--(3,0);
\draw [](1.6,0)--(2.3,0);
\draw [](1.6,0)--(1.3,0);
\draw [](0.1,0)--(0.4,0);

\draw [](10.4,0)--(9.8,0);
\draw [](7,0)--(7.8,0);
\draw [](7.8,0)--(8.4,0);
\draw [](8.4,0)--(8.7,0);
\draw [](9.8,0)--(9.5,0);

\draw [](7,0)[fill]circle[radius=0.07];
\draw [](7.8,0)[fill]circle[radius=0.07];
\draw [](8.4,0)[fill]circle[radius=0.07];
\draw [](9.8,0)[fill]circle[radius=0.07];
\draw [](10.4,0)[fill]circle[radius=0.07];

\end{tikzpicture}
\caption{Construction of the graph $H=G+P_1+P_2$.}
 \label{fig:lemmaProof1}
\end{figure}


Now we prove that $(G, u, v, k)$ is a YES instance of PCFE if and only if $(H, k)$ is a YES instance of Path Contraction. Moreover, we show that if $(K, k)$ is the kernel we get for Path Contraction on $(H, k)$ according to the argument in Theorem \ref{theorem:PCkernel}, then $(K, u', v', k)$ is a kernel for PCFE on $(G, u, v, k)$.

On the one hand, it is obvious to see that if $(G, u, v, k)$ is a YES instance of PCFE, then we can do the same (at most $k$) edge contractions on $H$, which would result in a path with endvertices $u'$ and $v'$.

On the other hand, suppose $(H, k)$ is a YES instance of Path Contraction. Let $\Phi$ be {a minimum} set of edges contracted that modifies $H$ into a path $P$.  {A path of length $k+1$ will still be a nontrivial path after at most $k$ edge contractions.} As both $P_1$ and $P_2$ have length $k+1$, the path $P$ must have endvertices $s$ and $t$ where $s \in P_1$ and $t \in P_2$.  Actually we must have $s = u'$ and $t = v'$ by the minimality of $\Phi$. The path $P$ must pass through some $u_0$ and $v_0$,  {where $u \in W(u_0)$ and $v \in W(v_0)$.} It is easy to see that when we contract those edges in $\Phi \cap E(G)$ on $G$, we will make $G$ into a path between $u_0$ and $v_0$.

By the above argument, we know that  $(G, u, v, k)$ is a YES instance of PCFE if and only if $(H, k)$ is a YES instance of Path Contraction.
Now we show how to get a kernel for PCFE on $(G, u, v, k)$ from a kernel for Path Contraction on $(H, k)$. 

 {Let $(K, k)$ be a kernel for Path Contraction on $(H, k)$ according to the argument in Theorem \ref{theorem:PCkernel}. Since both the lengths of $P_1$ and $P_2$ in $H$ are $k+1$, $|V(P_{i})| < k+3$, no edge on them satisfies the condition in Rule B. So the kernelization does not contract any edge in $P_1$ and $P_2$, we must have $u', v' \in V(K)$. Moreover, it is not hard to see that $(K, k)$ is a YES instance for Path Contraction if and only if $(K, u',v', k)$ is a YES instance for PCFE, thus $(K, u',v',k)$ is a kernel for PCFE on  $(G, u, v, k)$. Since $|V(K)| \leq 3k+4$, we get a kernel for PCFE on  $(G, u, v, k)$ as we want.}
\end{proof}

%
%
%
%

Now we are ready to prove our kernel bound for the Cycle Contraction problem.
We adopt the following reduction rules from \cite{belmonte2013parameterized}.

\textbf{Reduction Rule 1} \cite{belmonte2013parameterized}
If $G$ is 3-connected and $|V(G)| > 2k+4$, then return NO.

\textbf{Reduction Rule 2} \cite{belmonte2013parameterized}
If $G$ contains a block $B$ on at least $k+2$ vertices and $V(G)\setminus V(B) \neq \emptyset$, then return NO if $|V(G)\setminus V(B)|\geq k+1$, and return the instance
$(G', k- |V(G)\setminus V(B)|)$ otherwise, where $G'$ is the graph obtained from $G$ by exhaustively contracting a vertex of $V(G)\setminus V(B)$ onto one of its neighbours.

\textbf{Reduction Rule 3} \cite{belmonte2013parameterized}
If $G$ contains a block $B$ on at most $k+1$ vertices and $|V(G)\setminus V(B)|\geq k+2$, then return the instance $(G', k- |V(B)|)$, where $G'$ is the graph
obtained from $G$ by exhaustively contracting a vertex of $V(B)$ onto one of its neighbours.

Note that any connected graph that is not a tree can be contracted to a cycle. We call a cycle $C$ \textit{optimal} for $G$ if $C$ is the longest cycle to which $G$ can be contracted. 

\begin{lemma} \cite{belmonte2013parameterized}\label{lemma2}
Let $(G, k)$ be a YES instance of Cycle Contraction, $C$ be an
optimal cycle for $G$, and $W$ be a $C$-witness structure of $G$. If $G$ is 2-connected
and contains two vertices $u$ and $v$ such that $d_G(u) = d_G(v) = 2$ and $G - \{u, v\}$ has exactly two connected components $G_1$ and $G_2$, then the following three statements hold:
 
\begin{enumerate}
 \item  Either $\{u\}$ and $\{v\}$ are small witness sets of $W$, or $u$ and $v$ belong to the
same big witness set of $W$;
\item If $u$ and $v$ belong to the same big witness set $W \in W$, then $W$ contains all
the vertices of $G_1$ or all the vertices of $G_2$;
\item If $G_1$ and $G_2$ contain at least $k + 1$ vertices each, then $\{u\}$ and $\{v\}$ are
small witness sets of $W$.

\end{enumerate}
 
\end{lemma}

Based on the observation in Lemma \ref{lemma2}, we introduce a novel reduction rule which is the key to make the improvement.

\textbf{Reduction Rule 4}
Let $(G,k)$ be an instance of Cycle Contraction, where $G$ is 2-connected. If $G$ contains two vertices $x$ and $y$ such that $d_G(x) = d_G(y) = 2$, and the graph $G - \{x, y\}$ has exactly two connected components $G_1$ and $G_2$, such that  $|V(G_1)| \geq k+2$ and $|V(G_2)| \geq k+2$. Then we can obtain a kernel $(K,k)$ for $(G,k)$ with $|V(K)|\leq 4.5k+10$.

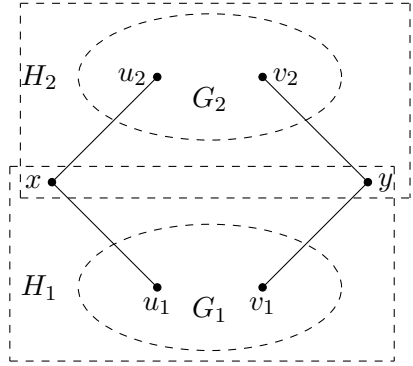
\begin{figure}
\centering
\begin{tikzpicture}[scale = 0.7]

\draw [dashed](5, 0) ellipse (2.5 and 1.2);
\draw [dashed](5, 4) ellipse (2.5 and 1.2);

\draw [](2.3,0)node[left]{$H_1$};
\draw [](2.3,4)node[left]{$H_2$};
\draw [](2,2)node[left]{$x$};
\draw [](5,0)node[below]{$G_1$};
\draw [](8,2)node[right]{$y$};
\draw [](5,4)node[below]{$G_2$};

\draw [](4,4)node[left]{$u_2$};
\draw [](4,0)node[below]{$u_1$};
\draw [](6,0)node[below]{$v_1$};
\draw [](6,4)node[right]{$v_2$};

\draw [](2,2)[fill]circle[radius=0.07];
\draw [](8,2)[fill]circle[radius=0.07];
\draw [](4,4)[fill]circle[radius=0.07];
\draw [](6,4)[fill]circle[radius=0.07];
\draw [](4,0)[fill]circle[radius=0.07];
\draw [](6,0)[fill]circle[radius=0.07];

\draw [](2,2)--(4,4);
\draw [](8,2)--(6,4);
\draw [](2,2)--(4,0);
\draw [](8,2)--(6,0);

\draw [dashed](1.2,2.3)--(8.5,2.3);
\draw [dashed](1.2,-1.4)--(8.5,-1.4);
\draw [dashed](1.2,2.3)--(1.2,-1.4);
\draw [dashed](8.5,2.3)--(8.5,-1.4);

\draw [dashed](1.4,1.7)--(8.8,1.7);
\draw [dashed](1.4,5.4)--(8.8,5.4);
\draw [dashed](1.4,1.7)--(1.4,5.4);
\draw [dashed](8.8,5.4)--(8.8,1.7);

\end{tikzpicture}
\caption{Constructions of $H_1$ and $H_2$.}
 \label{fig:reductionProof}
\end{figure}


We now prove the correctness of Reduction Rule 4.

\begin{lemma}
Reduction Rule 4 is safe. 
\end{lemma}

\begin{proof}
Let's construct two graphs $H_1 = G - V(G_2)$, $H_2 = G- V(G_1)$, see Figure \ref{fig:reductionProof} for an illustration. Both $x$ and $y$ have degree 2, let $N_{G}(x) = \{u_1, u_2\}$ and $N_{G}(y) = \{v_1, v_2\}$ with $\{u_1, v_1\} \subseteq V(G_1), \{u_2, v_2\} \subseteq V(G_2)$. Since $G$ is 2-connected, we know $u_1 \neq v_1$ and $u_2 \neq v_2$.
By statement 3 in Lemma \ref{lemma2}, we know that both $\{x\}$ and $\{y\}$ should be small witness sets, thus the problem of contracting $G$ into a cycle is equivalent to doing at most $k$ edge contractions that make both $H_1$ and $H_2$ into paths between $x$ and $y$. Note that one can contract $H_i$ into a path between $x$ and $y$ where both $\{x\}$ and $\{y\}$ are small witness sets with at most $k_i$ edge contractions  if and only if  $(G_i, u_i, v_i, k_i)$ is a YES instance for PCFE with $i\in \{1, 2\}$.

Let's consider the following two instances of PCFE, $(G_1, u_1, v_1, \lfloor{k/2}\rfloor)$ and $(G_2, u_2, v_2, \lfloor{k/2}\rfloor)$. If the answers to both $(G_1, u_1, v_1,$ $ \lfloor{k/2}\rfloor)$ and $(G_2, u_2, v_2, $ $ \lfloor{k/2}\rfloor)$ are NO, then we know that $(G, k)$ is a No instance for Cycle Contraction.

If either enquiry gives a kernel, then we know it should have at most $3\lfloor{k/2}\rfloor+4$ vertices by Theorem \ref{theorem:PCFE}. Without loss of generality, assume that $(G_1, u_1, v_1, $ $ \lfloor{k/2}\rfloor)$ gives a kernel  $(K_1, u'_{1}, v'_{1}$, $\lfloor{k/2}\rfloor)$ following the argument in Theorem \ref{theorem:PCFE}. We further look at $(G_2, u_2, v_2, k)$ as an instance of PCFE. If the answer to $(G_2, u_2, v_2, k)$ is NO, then we know that $(G, k)$ is also a No instance for Cycle Contraction. Otherwise, we get a kernel $(K_2, u'_{2}, v'_{2}, k)$ for $(G_2, u_2, v_2, k)$ with at most $3k+4$ vertices by Theorem \ref{theorem:PCFE}. 

 {Let $R$ be the graph obtained from $K_1$ and $K_2$ by adding edges $\{xu'_{1}, xu'_{2}\}$ and $\{yv'_{1}, yv'_{2}\}$, which implies that $|V(R)|=|V({K_1})|+|V(K_{2})|+|\{x,y\}|\leq 3\lfloor{k/2}\rfloor+4 + 3k+4 +2\leq 4.5k+10$. } 

\textbf{Claim:} $(R,k)$ is a kernel for Cycle Contraction on $(G, k)$. 

Observe that $G$ can be contracted into a cycle by at most $k$ edge contractions, if and only if there exist two non-negative integers $k_1$ and $k_2$ such that $k_1 + k_2 =k$, and $(G_i, u_i, v_i, k_i)$ is a YES instance for PCFE with $i\in \{1, 2\}$. Recall that $(K_1, u'_{1}, v'_{1}$, $\lfloor{k/2}\rfloor)$ is a kernel for $(G_1, u_1, v_1, \lfloor{k/2}\rfloor)$, and the two pendent paths $P_1$ and $P_2$ in $K_1$ constructed according to the argument of Theorem \ref{theorem:PCFE} have total length $2(\lfloor{k/2}\rfloor+1) > k_1$. So $(G_1, u_1, v_1, k_1)$ is a YES instance for PCFE if and only if $K_1$ can be contracted into a path between $u'_{1}$ and $v'_{1}$ by at most $k_1$ edge contractions. And $(K_2, u'_{2}, v'_{2}, k)$ is a kernel for $(G_2, u_2, v_2, k)$, $(G_2, u_2, v_2, k_2)$ is a YES instance for PCFE if and only if $K_2$ can be contracted into a path between $u'_{2}$ and $v'_{2}$ by at most $k_2$ ($\leq k$) edge contractions. Thus $(R,k)$ is a kernel for Cycle Contraction on $(G, k)$.
\end{proof}

\begin{theorem}
The Cycle Contraction problem admits a kernel with at most $5k+4$ vertices.
\end{theorem}
\begin{proof}
 We describe our algorithm to obtain the claimed kernel for the Cycle Contraction problem. The correctness of the algorithm follows from the correctness of the reduction rules. Given an instance $(G, k)$ of Cycle Contraction, the algorithm begins with
 applications of Reduction Rules 1-4 we listed above. Let $K$ be the resulting instance after all possible applications of the reduction rules.
 If $K$ is 3-connected, then we must have that $|V(K)| \leq 2k+3$, as otherwise Reduction Rule 1 could be applied.

If $K$ is not 2-connected, then $K$ has at least two blocks as $K$ is connected by assumption. Let $B$ be any block of $K$. Then $|V(B)| \leq k + 1$, as otherwise Rule 2 could be applied. Moreover, $|V(K) \setminus V(B)| \leq k + 1$ due to the assumption that Rule 3 cannot be applied. Hence $|V(K)| \leq 2k + 2$. In the following, we assume that $K$ is 2-connected and will prove the following claim.

 \textbf{Claim:} If $K$ is 2-connected and $|V(K)|\geq 5k+4$, then we can safely return NO.
 
 To see why the claim is correct, suppose, on the contrary, that $(K,k)$ is a YES instance of Cycle Contraction such that $K$ is a 2-connected graph with at least $5k+4$ vertices after all applications of Reduction Rules 1-4. Let $C$ be an optimal cycle for $K$ and let $W$ be a $C$-witness structure of $K$. We know that $V(K) = V_1 \cup V_2 \cup V_3$, where $V_1 = \{v\in W_{b}(u)| u\in C\}$, $V_{2} = \{ u| d_{C}(u) = 2, d_{K}(u)=2 \} $ and $V_{3} = \{u| d_{C}(u) = 2, d_{K}(u)>2 \}$. There are at most $k$ big witness sets in $W$, which in total contains at most $2k$ vertices, thus $|V_1| \leq 2k$.  Any vertex $v\in  V_3$ must be adjacent to vertices in $V_1$ as $d_{K}(v) > d_{C}(v)=2$ and it is in a small witness set, so $|V_3|\leq 2k$. Thus we know $|V_2|=|V(K)|-|V_1|-|V_3| \geq 5k+4 - 2k -2k = k+4$. Let's call any vertex in $V_2$ \textit{a small-witness vertex}. In the following, we will prove that there must be two small-witness vertices $x$ and $y$ such that  $K-\{x, y\}$ has two components each of which contains at least $k+2$ vertices, thus Reduction Rule 4 can be applied, a contradiction.

 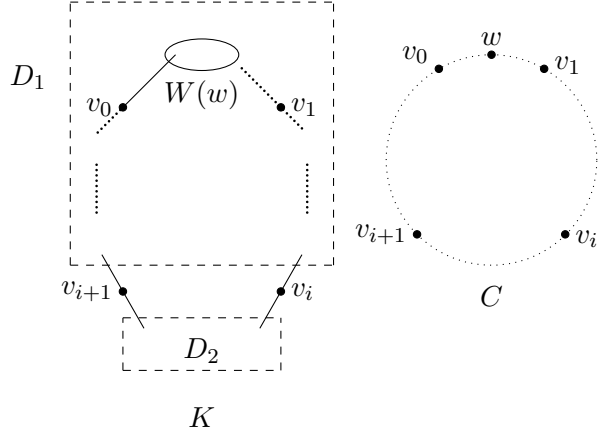
\begin{figure}\label{fig:thmProof}
 \centering
 \begin{tikzpicture}[scale = 0.7]
 
 \draw [](2.5, 6) ellipse (0.7 and 0.3);

 \draw [dashed](1,0)--(4,0);
 \draw [dashed](1,0)--(1,1);
 \draw [dashed](4,0)--(4,1);
 \draw [dashed](1,1)--(4,1);
 \draw [dashed](0,2)--(5,2);
 \draw [dashed](0,7)--(5,7);
 \draw [dashed](0,2)--(0,7);
 \draw [dashed](5,2)--(5,7);
 \draw [](1.4, 0.8)--(0.6,2.2);
 \draw [](3.6, 0.8)--(4.4,2.2);
 \draw [line width=1pt, line cap=round, dash pattern=on 0pt off 2\pgflinewidth](0.5,3)--(0.5,4);
 \draw [line width=1pt, line cap=round, dash pattern=on 0pt off 2\pgflinewidth](4.5,3)--(4.5,4);
 
 \draw [](1, 5)--(2,6);
 \draw [line width=1pt, line cap=round, dash pattern=on 0pt off 2\pgflinewidth](4.4, 4.6)--(3.2, 5.8);
 \draw [line width=1pt, line cap=round, dash pattern=on 0pt off 2\pgflinewidth](0.95, 4.95)--(0.5, 4.5);

 \draw [dotted](8, 4)[]circle[radius=2];

 \draw [](1,1.5)[fill]circle[radius=0.07];
 \draw [](4,1.5)[fill]circle[radius=0.07];
 \draw [](1,5)[fill]circle[radius=0.07];
 \draw [](4,5)[fill]circle[radius=0.07];

 \draw [](1,1.5)node[left]{$v_{i+1}$};
 \draw [](4,1.5)node[right]{$v_{i}$};
 \draw [](1,5)node[left]{$v_{0}$};
 \draw [](4,5)node[right]{$v_{1}$};
 \draw [](2.5,5.7)node[below]{$W(w)$};
 
 \draw [](8,6)[fill]circle[radius=0.07];
 \draw [](8,6)node[above]{$w$};
 \draw [](9.4,2.586)[fill]circle[radius=0.07];
 \draw [](9.4,2.5)node[right]{$v_{i}$};
 \draw [](6.586,2.586)[fill]circle[radius=0.07];
 \draw [](6.586,2.586)node[left]{$v_{i+1}$};
 \draw [](9,5.732)[fill]circle[radius=0.07];
 \draw [](9,5.8)node[right]{$v_{1}$};
 \draw [](7,5.732)[fill]circle[radius=0.07];
 \draw [](7,6)node[left]{$v_{0}$};
 
 \draw [](8,1.8)node[below]{$C$};
 \draw [](2.5,-0.5)node[below]{$K$};
 \draw [](2.5,0.8)node[below]{$D_2$};
 \draw [](-0.8,6)node[below]{$D_1$};

 \end{tikzpicture}
 \caption{The kernel $K$ and its optimal cycle $C$.}  
 \end{figure}

  Choose a small-witness vertex $v_{0}$, and let $w$ be the neighbour of $v_{0}$ clockwisely in $C$, as shown in Figure \ref{fig:thmProof}. We want to find another small-witness vertex $v$ in the cycle $C$ such that $K-\{v_0, v\}$ contains two connected components, each of which has size at least $k+2$.  Starting at $v_0$, let's look at the small-witness vertices in $C$ one by one clockwisely. And denote these vertices by $v_{i}$ with $i = 0,1,2, \ldots, |V_2|$. Let $i \geq 0$ be the smallest subscript such that the component containing $W(w)$ in $K-\{v_0,v_{i+1}\}$ contains at least $k+2$ vertices. If $K-\{v_0, v_{i+1}\}$ contains two connected components each of which has size at least $k+2$, then we are done. Otherwise, we know the number of vertices in the component not containing $W(w)$ in $K-\{v_0,v_{i+1}\}$ is less than $k+2$. Denote the two components of $K-\{v_{i},v_{i+1}\}$ by $D_1$ and $D_2$, where $W(w) \subseteq V(D_1)$. Since $|V(K)|\geq 5k+4$ and the component in $K-\{v_0,v_{i}\}$ containing $W(w)$ has less than $k+2$ vertices by the choice of $i$, we have $|V(D_2)| \geq  5k+4-2(k+1)-3=3k-1\geq k+2$. As there is no small-witness vertex in $D_2$, there are at least $k+4- |\{v_{i},v_{i+1}\}|=k+2$ small-witness vertices in $D_1$, thus $|V(D_1)|\geq k+2$. Let $x = v_i$ and $y=v_{i+1}$, then $K-\{x, y\}$ has two components each of which contains at least $k+2$ vertices. Thus we have found two vertices with the requested properties.
 
 It remains to observe that our kernelization algorithm can be run in polynomial time. For Reduction Rule 1, it takes $O(n^3)$ steps to check if a graph is 3-connected. And for Reduction Rule 2 or 3, it takes $O(n^2)$ steps to decide if they are applicable. And each application of Reductions 2 and 3 either returns NO or decreases the number of vertices, they can be exhaustively applied in polynomial time. As for Reduction Rule 4, note that the kernelization for PCFE can be run in polynomial time, since the kernelization for Path Contraction can be applied in polynomial time. Thus Reduction Rule 4 can also be applied in polynomial time, by simply checking all possible pairs of vertices with degree 2 in the graph to see if we need to apply the kernelization for PCFE.
\end{proof}


\section{Conclusion}
In the past decade, much effort has been put into obtaining better parameter dependence in the running time for all kinds of classical parameterized problems, like \textsc{Vertex Cover}, \textsc{Feedback Vertex Set}, \textsc{Multiway Cut} and so on. There are mainly two directions of algorithmic improvement for a problem that has been proved to be FPT, to obtain a better running time and to obtain a better kernel. In this paper, we provide a kernel for the Cycle Contraction problem with at most $5k+4$ vertices, which is a non-trivial improvement over the $6k+6$ kernel in \cite{belmonte2013parameterized}. Our improvement relies on observing the connection between Path Contraction and Cycle Contraction, which allows us to utilize an existing result on Path Contraction problem. 

In directed graphs, there are two types of contractions, i.e. the set contraction and path contraction, see the definitions in \cite{bang2009digraphs}.
It would be interesting to see whether the paramterized tractability results of the contraction problems can be generalized to the directed case.

\section{Acknowledgement}
The research of Bin Sheng was partially supported by China Scholarship Council (No. 201306140052). 
Yuefang Sun was supported by
National Natural Science Foundation of China (No. 11401389), China
Scholarship Council (No. 201608330111) and Zhejiang Provincial Natural Science Foundation of China (No. LY17A010017).

\bibliographystyle{plain}
\bibliography{references}

\begin{thebibliography}{10}

\bibitem{bang2009digraphs}
J{\o}rgen Bang-Jensen and Gregory~Z Gutin.
\newblock {\em Digraphs: theory, algorithms and applications}.
\newblock Springer Science \& Business Media, 2 edition, 2009.

\bibitem{belmonte2013parameterized}
R{\'e}my Belmonte, Petr~A Golovach, Pim van't Hof, and Dani{\"e}l Paulusma.
\newblock Parameterized complexity of two edge contraction problems with degree
  constraints.
\newblock In {\em International Symposium on Parameterized and Exact
  Computation}, pages 16--27. Springer, 2013.

\bibitem{bollobas2013modern}
B{\'e}la Bollob{\'a}s.
\newblock {\em Modern graph theory}, volume 184.
\newblock Springer Science \& Business Media, 2013.

\bibitem{cygan2015parameterized}
Marek Cygan, Fedor~V Fomin, {\L}ukasz Kowalik, Daniel Lokshtanov, D{\'a}niel
  Marx, Marcin Pilipczuk, Micha{\l} Pilipczuk, and Saket Saurabh.
\newblock {\em Parameterized Algorithms}, volume~4.
\newblock Springer, 2015.

\bibitem{downey2013fundamentals}
Rodney~G Downey and Michael~R Fellows.
\newblock {\em Fundamentals of parameterized complexity}, volume~4.
\newblock Springer, 2013.

\bibitem{Flum:2006:PCT:1121738}
J.~Flum and M.~Grohe.
\newblock {\em Parameterized Complexity Theory (Texts in Theoretical Computer
  Science. An EATCS Series)}.
\newblock Springer-Verlag New York, Inc., Secaucus, NJ, USA, 2006.

\bibitem{golovach2011increasing}
Petr~A Golovach, Marcin Kami{\'n}ski, Dani{\"e}l Paulusma, and Dimitrios~M
  Thilikos.
\newblock Increasing the minimum degree of a graph by contractions.
\newblock In {\em International Symposium on Parameterized and Exact
  Computation}, pages 67--79. Springer, 2011.

\bibitem{golovach2013obtaining}
Petr~A Golovach, Pim van't Hof, and Dani{\"e}l Paulusma.
\newblock Obtaining planarity by contracting few edges.
\newblock {\em Theoretical Computer Science}, 476:38--46, 2013.

\bibitem{guillemot2013faster}
Sylvain Guillemot and D{\'a}niel Marx.
\newblock A faster fpt algorithm for bipartite contraction.
\newblock {\em Information Processing Letters}, 113(22):906--912, 2013.

\bibitem{guo2015obtaining}
Chengwei Guo and Leizhen Cai.
\newblock Obtaining split graphs by edge contraction.
\newblock {\em Theoretical Computer Science}, 607:60--67, 2015.

\bibitem{heggernes2013obtaining}
Pinar Heggernes, Pim~van't Hof, Daniel Lokshtanov, and Christophe Paul.
\newblock Obtaining a bipartite graph by contracting few edges.
\newblock {\em SIAM Journal on Discrete Mathematics}, 27(4):2143--2156, 2013.

\bibitem{heggernes2014contracting}
Pinar Heggernes, Pim van't Hof, Benjamin L{\'e}v{\^e}que, Daniel Lokshtanov,
  and Christophe Paul.
\newblock Contracting graphs to paths and trees.
\newblock {\em Algorithmica}, 68(1):109--132, 2014.

\bibitem{kratsch2014recent}
Stefan Kratsch.
\newblock Recent developments in kernelization: A survey.
\newblock {\em Bulletin of EATCS}, 2(113), 2014.

\bibitem{li2016improved}
Wenjun Li, Qilong Feng, Jianer Chen, and Shuai Hu.
\newblock Improved kernel results for some fpt problems based on simple
  observations.
\newblock {\em Theoretical Computer Science}, 2016.

\bibitem{lokshtanov2012kernelization}
Daniel Lokshtanov, Neeldhara Misra, and Saket Saurabh.
\newblock Kernelization-preprocessing with a guarantee.
\newblock In {\em The Multivariate Algorithmic Revolution and Beyond}, pages
  129--161. Springer, 2012.

\bibitem{lokshtanov2013hardness}
Daniel Lokshtanov, Neeldhara Misra, and Saket Saurabh.
\newblock On the hardness of eliminating small induced subgraphs by contracting
  edges.
\newblock In {\em International Symposium on Parameterized and Exact
  Computation}, pages 243--254. Springer, 2013.

\bibitem{mathieson2012editing}
Luke Mathieson and Stefan Szeider.
\newblock Editing graphs to satisfy degree constraints: A parameterized
  approach.
\newblock {\em Journal of Computer and System Sciences}, 78(1):179--191, 2012.

\end{thebibliography}

\end{document}